\documentclass[letterpaper]{article} 
\usepackage[pass]{geometry}
\usepackage{graphicx}  
\usepackage{float}
\usepackage{amsmath}
\usepackage{amssymb}
\usepackage{amsthm}
\usepackage{subcaption}
\newcommand{\R}{\mathcal{R}}
\newcommand{\leaveout}[1]{}
\usepackage[noend]{algpseudocode}
\theoremstyle{definition}

\newtheorem{prop}{Proposition}

\usepackage{mathtools}

\DeclarePairedDelimiter\floor{\lfloor}{\rfloor}

\floatstyle{ruled}
\newfloat{algorithm}{tbp}{loa}
\providecommand{\algorithmname}{Algorithm}
\floatname{algorithm}{\protect\algorithmname}

\usepackage{algorithm}
\usepackage{algpseudocode}

\usepackage[colorinlistoftodos]{todonotes}
\title{Adversarial Task Allocation}
\author{Chen Hajaj and Yevgeniy Vorobeychik\\
Electrical Engineering and Computer Science\\
Vanderbilt University\\
Nashville, TN\\
chen.hajaj,yevgeniy.vorobeychik@vanderbilt.edu}

\begin{document}
\maketitle
\begin{abstract}
The problem of allocating tasks to workers is of long standing
fundamental importance.
Examples of this include the classical problem of assigning computing
tasks to nodes in a distributed computing environment, as well as the more recent problem of crowdsourcing
where a broad array of tasks are slated to be completed by human workers.
Extensive research into this problem generally addresses important
issues such as uncertainty and, in crowdsourcing, incentives.
However, the problem of adversarial tampering with the task allocation
process has not received as much attention.

We are concerned with a particular adversarial setting in task
allocation where an attacker may target a specific worker in order to
prevent the tasks assigned to this worker from being completed.
We consider two attack models: one in which the adversary observes
only the allocation policy (which may be randomized), and the second
in which the attacker observes the actual allocation decision.
For the case when all tasks are homogeneous, we provide
polynomial-time algorithms for both settings.
When tasks are heterogeneous, however, we show the adversarial
allocation problem to be NP-Hard, and present algorithms for solving
it when the defender is restricted to assign only a single worker per
task.
Our experiments show, surprisingly, that the difference between the
two attack models is minimal: deterministic allocation can achieve
nearly as much utility as randomized.

\end{abstract}

\section{Introduction}

The problem of allocating a set of tasks among a collection of workers has been a fundamental research question in a broad array of domains, including
distributed computing, robotics, and, recently, crowdsourcing~\cite{alistarh2012allocate,stone1999task,liu2017sequential}.
Despite the extensive interest in the problem, however, there is little prior work on task allocation in settings where workers may be attacked, and their ability to successfully complete the assigned task compromised as a consequence.
Such \emph{adversarial task allocation} problems can arise, for example, when tasks are of high economic or political consequence, such as when we use crowdsourcing to determine which executables are malicious or benign, or which news stories constitute fake news.

We investigate the adversarial task allocation problem in which a rational attacker targets a single worker after tasks have already been assigned.
We consider two models of information available to the attacker at the time of the attack: \emph{partial information}, where the attacker only knows the defender's policy (common in Stackelberg security games, for example), and \emph{full information}, where the attacker observes the actual task assignment decision.
We formalize the interaction between the attacker and requester (defender) as a Stackelberg game in which the defender first chooses an allocation policy, and the attacker subsequently attacks a single worker so as to maximize the defender's losses from the attack.
We seek a strong Stackelberg equilibrium (SSE) of this game.
In the partial information setting, we study how to compute an optimal randomized assignment in an SSE for the requester, while in the full information model we focus on computing an optimal deterministic assignment.

We first consider a homogeneous task setting, where all tasks have the same utility.
In this case, we show that an optimal randomized task assignment policy can be computed in linear time.
Deterministic assignment is harder, and the algorithm we devise for that case only runs in pseudo-polynomial time (linear in the number of tasks, and quadratic in the number of workers).
While randomized policies in Stackelberg games are always advantageous for the defender (if they can be used), our experiments show that the difference between optimal deterministic and randomized policies is small, and shrinks as we increase the number of tasks, suggesting that deterministic policy may be a good option especially when we are uncertain about what the attacker can observe.

Next, we turn to heterogeneous tasks settings.
This case, it turns out, is considerably more challenging.
Nevertheless, if we impose a restriction that only a single worker can be assigned to a task (optimal when tasks are homogeneous, but not in general), we can still compute an optimal randomized assignment in linear time.
Optimal deterministic assignment is much harder even with this restriction in place, and we propose an integer programming approach for solving it.

\paragraph{Related Work}

The problem of task allocation in adversarial settings has been considered from several perspectives.
One major stream of literature is about robots acting in adversarial environments \cite{alighanbari2005cooperative,jones2006dynamically}.
Alighanbari and How~\cite{alighanbari2005cooperative} consider assigning weapons to targets, somewhat analogous to our problem, but do not model the decision of the adversary; their model also has rather different semantics than ours.
Robotic soccer is another common adversarial planning problem, although the focus is typically on coordination among multiple robots when two opposing teams are engaged in coordination and planning~\cite{jones2006dynamically}.

Another major literature stream which considers adversarial issues is crowdsourcing.
One class of problems is the issue of individual worker incentives in truthfully responding to questions~\cite{singla2013truthful}, or in the amount of effort they devote to the task~\cite{tran2014budgetfix,liu2017sequential}, rather than adversarial reasoning per se.
Another, more directly adversarial, considers situations where some workers simply answer questions in an adversarial way~\cite{Ghosh11,Steinhardt16}.
However, the primary interest in this work is robust estimation when tasks are assigned randomly or exogenously, rather than task assignment itself.
Similarly, prior research on machine learning when a portion of data is adversarially poisoned~\cite{chen-icml11,xu2010,feng2014robust,chen2013robust} focuses primarily on the robust estimation problem, and not task allocation; in addition, it does not take advantage of structure in the data acquisition process, where workers, rather than individual data points, are attacked.


Our work has a strong connection to the literature on Stackelberg security games~\cite{conitzer2006computing,korzhyk2010complexity,tambe2011security}.
However, the mathematical structure of our problem is different: for example, we have no protection resources to allocate, and instead the defender's decision is about allocating tasks to potentially untrusted workers.


\section{Model}
\label{sec:model}


Consider an environment populated with a single requester (hereafter denoted ``\textit{defender}''), a set of $n$ workers, $W$, a set of $m$ binary labeling tasks, $T$, and an \emph{adversary}. 
Each worker $w \in W$ is characterized by an individual proficiency, or the probability of successfully completing a binary labeling task, denoted $p_{w}$, and assume that $p_w \ge 0.5$ for all workers (otherwise, we can always flip the received labels).
In our setting, these proficiencies are known to the defender.\footnote{The issue of learning such proficiencies from experience has itself been extensively studied~\cite{sheng2008get,dai2011artificial,manino2016efficiency}.}
Further, we assume that $m$ is sufficiently small that any worker can complete all tasks.
For exposition purposes, we index the workers by integers $i$ in decreasing order of their proficiency, so that $P=\{p_1,\ldots,p_n\} $ s.t. $ p_{i}\geq p_{j}$ $\forall i<j$, where the set of k most proficient workers is defined as: $ W^k $.
Each task $t \in T$ is associated with a utility $u_{t}$ that the defender obtains if this task is completed correctly.
We assume that if the task is not completed, or incorrect, the defender obtains zero utility from it.
Let $\iota_t$ be the (unknown) correct label corresponding to a task $t$.
%

The defender's fundamental decision is the \emph{assignment} of tasks to workers.
Formally, an assignment $s$ specifies a subset of tasks $T'(s)$ and the set of workers, $W_t(s)$ assigned to each task $t \in T'(s)$.
Let $L_t(s)$ denote the labels returned by workers in $W_t(s)$ for $t$.
Suppose that the defender faces a budget constraint, of assigning $m$ tasks; thus, each task can be assigned to a single worker, or a subset of tasks assigned to multiple workers.\footnote{If there are more tasks than budget, we can simply take the $m$ tasks with the highest utility.}
Then the defender determines the final label to assign to $t$ according to some deterministic mapping $ \delta: L_t \rightarrow l $ (e.g., majority label), such that $ L\in \{0,1\}^{|W_t(s)|} $ and $ l\in\{0,1\} $.
We assume that whenever a single worker $w$ is a assigned to a task and returns a label $l_w$, $\delta(l_w) = l_w$.
The defender's expected utility when assigning a set of tasks $T'(s)$ to workers and obtaining the labels is then
\begin{equation}
u_{def}(s)=\sum_{t\in T'(s)} u_{t} \mathbb{E}[\mathbb{I}\{\delta(L_t(s)) = \iota_t\}],
\label{eq:basic}
\end{equation}
where $\mathbb{I}\{\cdot\}$ is an indicator function and the expectation is with respect to labeler proficiencies (and resulting stochastic realizations of labels).

It is immediate that in our setting, if there is no adversary, all tasks should be assigned to the worker with the highest $p_w$.
Our focus, however, is how to optimally assign workers to tasks when there is an intelligent adversary who could subsequently (to the assignment) attack one of the workers.
In particular, we assume that there is an adversary (attacker) with the goal of minimizing the defender's utility $u_{def}$; thus, the game is zero sum.
To this end, the attacker chooses a single worker to attack, for example, by deploying a cyber attack against the corresponding compute node, or against the device on which the human worker performs the tasks assigned to them.
We encode the attacker's strategy by a vector $\alpha$ where $\alpha_w = 1$ iff a worker $w$ is attacked (and $\sum_w \alpha_w = 1$ since exactly one is attacked).
The attack takes place \emph{after} the tasks have already been assigned to workers.

We distinguish between two forms of adversary's \emph{knowledge} about worker-to-task assignment before deploying the attack: 1) \emph{partial knowledge}, when the adversary only knows the defender's policy (which may be deterministic or randomized), and 2) \emph{complete knowledge}, when the attacker knows the actual assignments of tasks to workers.
The specific consequences of the attack---denial of service, where the targeted node is taken offline and cannot communicate the labels to the defender, or integrity attack, where incorrect labels are reported---are immaterial in our model, since the defender receives zero utility from the tasks assigned to the attacked worker in either case.


Clearly, when an attacker is present, the policy of allocating all tasks to the most competent worker (or any other) will yield zero utility for the defender.
The challenge of how to split the tasks up among workers, trading off quality with robustness to attacks, becomes decidedly non-trivial.
Our goal is to address this challenge for both models of adversarial knowledge, computing an optimal randomized assignment (i.e., a probability distribution $q$ over assignments $s$) in the partial knowledge environment, and an optimal deterministic assignment $s$ in the complete knowledge setting.
Formally, we aim to compute a strong Stackelberg equilibrium of the game between the requester (leader), who chooses a task-to-worker assignment policy, and the attacker (follower), who attacks a single worker~\cite{stackelberg1952theory}.

\section{Homogeneous tasks}

We start by considering tasks which are \emph{homogeneous}, that is, $u_t = u_{t'}$ for any two tasks $t,t'$.
Without loss of generality, suppose that all $u_t = 1$.
Note that it is immediate that we never wish to waste budget, since assigning a worker always results in non-negative marginal utility.
Next we consider the problem of optimal randomized assignment when the attacker only knows the (randomized) policy, and optimal deterministic assignment (when the attacker observes the actual assignment), showing that both can be solved efficiently.


\subsection{Randomized strategy} 

In general, a randomized allocation involves a probability distribution over all possible matchings with cardinality $m$ between tasks and workers.
We first observe that this space can be narrowed to consider only matchings in which one worker is assigned to any task.
\begin{prop}\label{prop:oneEach} 
Suppose that tasks are homogeneous. 
There exists a Stackelberg equilibrium in which the defender commits to a randomized strategy with all assignments in the support assigning at most one worker per task.
\end{prop} 
\begin{proof} 
Consider an optimal randomized strategy commitment restricted to assign at most one worker per task, and the associated Nash equilibrium (which exists, by equivalence of Stackelberg and Nash in zero-sum games~\cite{Korzhyk11}).
We now show that this remains an equilibrium even in the unrestricted space of assignments for the defender.

We prove by contradiction.
Suppose that there is $s$ which assigns multiple workers for some tasks and is strictly better for the defender.
Consider an arbitrary attack $\alpha$ in the support of $R$.
Given $s$, suppose that there is some task $t$ assigned to $k \ge 2$ workers.
Since only $m$ assignments can be made, there must be $k-1$ tasks which are not assigned.
If any of these workers is attacked, then moving this worker to another task will not change the defender's utility.
Thus, w.l.o.g., suppose none of the workers are attacked, and consider moving $k-1$ of these to unassigned tasks; let this be $s'$.
Under $s$, the marginal utility of the $k$ workers completing their assigned task $t$ is at most $u_t$.
Under $s'$, the marginal utility of these workers is $\sum_{i=1}^k p_iu_t \ge u_t$, since $p_i \ge 0.5$.
Thus, $s'$ is weakly improving.
Since this argument holds for an arbitrary $\alpha$ in the support of $R$, the resulting $s'$ must also be weekly improving given $R$.
Since $s$ is a strict improvement on the original Nash equilibrium strategy of the defender, then $s'$ must be as well, which means that this could not have been a Nash equilibrium, leading to a contradiction.
The result then follows from the known equivalence between Nash and Stackelberg equilibria in zero-sum games.
\end{proof}
As a consequence of this proposition, it suffices to consider assignment policies (randomized or deterministic) in which each task is assigned to a single worker, and all $m$ tasks are assigned.
Since there are $m$ tasks, an assignment is then the split of these among the workers.
Consider the unit simplex in $\R^n$, $\Delta = \{x|\sum_w x_w = 1\}$ which represents how we split up tasks among workers.
It is then sufficient to consider the space of assignments $S$ where $x \in \Delta$ means that each worker receives $s_w = mx_w$ tasks, with the constraint that all $mx_w$ are integers; i.e., $S = \{mx|x \in \Delta, mx \in \mathbb{Z}_+\}$.

A randomized allocation, in general, is a probability distribution $q$ over the set of assignments $S$.
In principle, considering the problem of computing an optimal randomized task allocation is daunting: even for only $14$ workers and $14$ tasks there are over 20 million possible assignments in $S$.
We now observe that in fact we can restrict attention to a far more restricted space of \emph{unit assignments}, $\tilde{S} = m\{e_w\}_{w \in W} \subset S$, where $e_w$ is a unit vector which is $1$ in $w$th position and $0$ elsewhere; i.e., assigning a single worker to all tasks.
Let $\lambda$ denote a distribution over $\tilde{S}$.

\begin{prop}\label{prop:rep}
  For any distribution over assignments $q$ and attack strategy $\alpha$, there exists a distribution over $\tilde{S}$, $\lambda$, which results in the same utility.
\end{prop}
\begin{proof}
  Fix an attacker strategy $\alpha$.
  For any probability distribution $q$ over $S$, the expected utility of the defender is $u_{def}(q,\alpha) = \sum_{w\in W} p_w (1-\alpha_w) \sum_{s \in S} q_s s_w$, where $s_w$ is the number of tasks assigned to worker $w$.
The expected utility of the defender for a distribution $\lambda$ over $\tilde{S}$ is $u_{def}(\lambda,\alpha) = m\sum_{w\in W} p_w (1-\alpha_w) \lambda_w$.
Define $\lambda_w = \frac{1}{m}\sum_{s \in S} q_s s_w$.
It suffices to show that $\sum_{w\in W} \lambda_w = 1$.
This follows since $\sum_{w\in W} \sum_{s \in S} q_s s_w = \sum_{s \in S} q_s \sum_{w\in W} s_w = m\sum_{s \in S} q_s = m$ because $\sum_{w\in W} s_w = m$ and $q$ is a probability distribution.
\end{proof}
This result allows us to restrict attention to probability distributions over $\tilde{S}$.

Next, we make another important observation which implies that in an optimal randomized assignment the support of $\lambda$ must include the best $k$ workers for some $k$.
Below, we use $i$ as the rank of a worker in a decreasing order of proficiency.
\begin{prop}\label{prop:optRand}
	In an optimal randomized assignment $\lambda$, suppose that $\lambda_i > 0$ for $i > 1$.
Then there must be an optimal assignment in which $\lambda_{i-1} > 0$.
\end{prop}
\begin{proof}
It is useful to write the utility of the defender as $\sum_{t\in T}u_t(\sum_i \lambda_i p_i - \lambda_ap_a)$, where $a$ is the worker being attacked.
Suppose that $ \lambda $ is an optimal randomized assignment, and there exist some worker $ i$, s.t. $ \lambda_i>0 $ and $ \lambda_{i-1} = 0 $. 
Since $\lambda_ip_i > 0$, there is $\epsilon > 0$ such that $(\lambda_i - \epsilon)p_i > \epsilon p_{i-1}$.
First, suppose that some node $a \ne i$ is being attacked.
Thus, $\lambda_ap_a \ge \lambda_jp_j$ for all $j \ne a$ (by optimality of the attacker).
Consequently, after $\epsilon$ was removed from the probability of assigning to $i$, node $a$ is still attacked, and the defender receives a net gain of $\epsilon(p_{i-1} - p_i) \ge 0$.
Thus, if $\lambda$ was optimal, so is the new assignment.
Now, suppose that $a = i$.
Again, if $a$ is still being attacked after $\epsilon$ is moved to $i-1$, the defender obtains a non-negative net gain as above.
If instead this change results in some other $j \ne i$ now being attacked, the defender obtains another net gain of $(\lambda_i-\epsilon)p_i + \epsilon p_{i-1} - \lambda_j p_j = (\lambda_i p_i - \lambda_j p_j) + \epsilon(p_{i-1} - p_i) \ge 0$, by optimality condition of the attacker and the fact that $p_{i-1} \ge p_i$.
Again, if $\lambda$ was optimal, so is the new assignment.
\end{proof}
The final piece of structure we observe is that in an optimal randomized assignment the workers in the support must have the same utility for the adversary.
Define $W_t(\lambda) = \{w \in W|\lambda_w > 0\}$, i.e., the workers in the support of a strategy $\lambda$.
\begin{prop}\label{prop:balance}
There exists an optimal randomized assignment with $\lambda_w p_w = \lambda_{w'} p_{w'}$ for all $w,w' \in W_t(\lambda)$.
\end{prop}
\begin{proof}
Suppose that an optimal $\lambda$ has two workers $w,w' \in W_t(\lambda)$ with $\lambda_w p_w > \lambda_{w'} p_{w'}$.
Define $u_{max} = \max_{w \in W_t(\lambda)} \lambda_w p_w$ and $u_{min} = \min_{w \in W_t(\lambda)} \lambda_w p_w$.
Let $W_{max}$ be the set of maximizing workers (with identical marginal value to the attacker, $u_{max}$), and let $z$ be some minimizing worker.
Define $K = |W_{max}|$.
By optimality of the attacker, some $w \in W_{max}$ is attacked and by our assumption $u_{max} > u_{min}$.
For any $w \in W_{max}$, define $\bar{p}_w = 1/p_w$ and similarly let $\bar{p}_z = 1/p_z$.

First, suppose that $\bar{p}_z < \frac{1}{K-1} \sum_{w \in W_{max}} \bar{p}_w$.
Then there is some $w \in W_{max}$ with $\bar{p}_z < \bar{p}_w$ or, equivalently, $p_z > p_w$.
Then there exists $\epsilon > 0$ small enough so that if we change $\lambda_z$ to $\lambda_z + \epsilon$ and $\lambda_w$ to $\lambda_w-\epsilon$ the attacker does not attack $z$, and we gain $\epsilon (p_z - p_w)$ and either lose the same as before to the attack (if $K>1$) or lose less (if $K=1$).
Consequently, $\lambda$ cannot have been optimal, and this is a contradiction.

Thus, it must be that $\bar{p}_z \ge \frac{1}{K-1} \sum_{w \in W_{max}} \bar{p}_w$.
Suppose now that we move all of $\lambda_z$ from $z$ onto all workers in $W_{max}$, maintaining their utility to the attacker as constant (and thus the attacker does not change which worker is attacked).
For any worker $w \in W_{max}$, the resulting $\lambda_w' = \lambda_w p_w + C$, where $C = \epsilon_w p_w$.
Moreover, it must be that $\sum_{w \in W_{max}} \epsilon_w = \lambda_z$.
Since $\epsilon_w = C/p_w$, we can find that $C = \frac{\lambda_z}{\sum_{w \in W_{max}} \bar{p}_w}$.
Consequently, the defender's net gain from the resulting change is
\[
(K-1) \frac{\lambda_z}{\sum_{w \in W_{max}} \bar{p}_w} - \lambda_z p_z = \lambda_z \left(\frac{(K-1)}{\sum_{w \in W_{max}} \bar{p}_w} - p_z\right),
\]
which is non-negative since $\bar{p}_z \ge \frac{1}{K-1} \sum_{w \in W_{max}} \bar{p}_w$.
We can then repeat the process iteratively, removing any other workers $z$ in the support but not in $W_{max}$ to obtain a solution with uniform $\lambda_w' p_w$ for all $w$ with $\lambda_w' > 0$ which is at least as good as the original solution $\lambda$.
\end{proof}

Algorithm~\ref{alg:DoSRand} uses these insights for computing an optimal randomized assignment in linear time.
\begin{algorithm}[hbt]
	\caption{Randomized assignment} \label{alg:DoSRand} 
	\textbf{input:} The set of workers $W$, and their proficiencies $P$\\
	\textbf{return:} The optimal randomized policy $ \lambda^* $
	\begin{algorithmic}[1]
		\State $ u_{max} \leftarrow 0 $
        \State $ v \leftarrow \frac{1}{p_1} $
		\For {$ k \in \{2,\ldots,n\} $}
			\State $ v \leftarrow v + \frac{1}{p_k} $
                        \If{$\frac{k-1}{v} > u_{max}$}
                                \State $u_{max} \leftarrow \frac{k-1}{v}$
                                \State $k^* \leftarrow k$

			\EndIf
		\EndFor
                \For {$ i \in \{1,\ldots,k^*\}$}
                \State $ \lambda_i^* \leftarrow  \frac{1}{p_i v} $
                \EndFor\\
		\Return  $ \lambda^*$
	\end{algorithmic}
\end{algorithm} 
At the high level, it attempts to compute the randomized assignments for all possible $k$ most proficient workers who can be in the support of the optimal assignment, and then returns the assignment which yields the highest expected utility to the defender.
For a given $k$, we can find $\lambda_i$ directly for all $ i \in \{1,\ldots,k\}$: $\lambda_i = \frac{1}{p_i \sum_j \frac{1}{p_j}}$.
Consequently, the utility to the defender of an optimal randomized assignment for $k$ tasks is $\frac{k-1}{\sum_j \frac{1}{p_j}}$ (since one worker is attacked, and it doesn't matter which one).

\subsection{Deterministic strategy}
Next we consider the setting in which the attacker observes the actual task assignment, in which case the defender's focus becomes on computing an optimal deterministic assignment.
Recall that we use $s = \{s_1,\ldots,s_n\}$ to denote the number of tasks allocated to each worker.
Although the space of deterministic allocations is large, we now observe several properties of optimal deterministic assignments which allow us to devise a polynomial time algorithm for this problem.

Our first several results are similar to the observations we made for randomized assignments, but require different arguments.
\begin{prop}\label{prop:oneEachDeterministic} 
Suppose that tasks are homogeneous. 
For any optimal deterministic strategy $s$ there is a weakly utility-improving deterministic assignment $s'$ for the requester which assigns each task to a single worker.
\end{prop} 
\begin{proof}
Consider an optimal assignment $s$ and the corresponding best response by the attacker, $\alpha$, in which a worker $\bar{w}$ is attacked.
Let a task $\bar{t}$ be assigned to a set of workers $W_{\bar{t}}$ with $|W_{\bar{t}}| = k > 2$.
Then there must be another task $t'$ which is unassigned.
Now consider a worker $w \in W_{\bar{t}}$.
Since utility is additive, we can consider just the marginal utility of any worker $w'$ to the defender and attacker; denote this by $u_{w'}$.
Let $T_{w'}$ be the set of tasks assigned to a worker $w'$ under $s$.
Let $u_w = \sum_{t \in T_w} u_{wt}^M$, where $u_{wt}^M$ is the marginal utility of worker of $w$ towards a task $t$.
Clearly, $u_w \le u_{\bar{w}}$, since the attacker is playing a best response.

Suppose that we reassign $w$ from $\bar{t}$ to $t'$.
If $w = \bar{w}$, the attacker will still attack $w$ (since the utility of $w$ to the attacker can only increase), and the defender is indifferent.
If $w \ne \bar{w}$, there are two cases: (a) the attacker still attacks $\bar{w}$ after the change, and (b) the attacker now switches to attack $w$.
Suppose the attacker still attacks $\bar{w}$.
The defender's net gain is $p_{w} - u_{w\bar{t}}^M \ge 0$.
If, instead, the attacker now attacks $w$, the defender's net gain is $u_{\bar{w}} - u_w \ge 0$.
\end{proof}
Consequently, the strategy space $S$ defined above (where a single worker is assigned for any task) still suffices.


Given a deterministic assignment $s \in S$ and the attack strategy $\alpha$, the defender's expected utility is: 
\begin{equation}\label{eq:pure}
u_{def}(s,\alpha)=\sum_{w\in W}s_{w} p_{w}(1-\alpha_w).
\end{equation}
We now derive a similar property of optimal deterministic assignments that held for randomized assignments: there is always an optimal deterministic assignment in which we assign the $k$ most proficient workers for some $k$.

\begin{prop}\label{th:optDet}
	In an optimal deterministic assignment $s$, suppose that $s_i > 0$ for $i > 1$.
Then there must be an optimal assignment in which $s_{i-1} > 0$.
\end{prop}
\begin{proof}
Consider an optimal deterministic assignment $s$ and the attacker's best response $\alpha$ in which a worker $w$ is attacked.
Now, consider moving 1 task from $i-1$ to $i$.
Suppose that $w = i$, that is, the worker $i$ is attacked.
If the change results in $i-1$ being attacked, the net gain to the defender is $p_i (|T_i|-1) \ge 0$.
Otherwise, the net gain is $p_{i-1} > 0$.
Suppose that another worker $w \ne i$ is attacked.
If $i-1$ is now attacked, the net gain is $p_w (|T_w|-1) \ge 0$
Otherwise, the net gain is $p_{i-1} - p_i \ge 0$.
\end{proof}

Next, we present an allocation algorithm for optimal deterministic assignment (Algorithm~\ref{alg:deterministic}) which has complexity $O(n^2m)$, quadratic in the number of workers and linear in the number of tasks.
The intuition behind the algorithm is to consider each worker $i$ as a potential target of an attack, and then compute the best deterministic allocation subject to a constraint that $i$ is attacked (i.e., that $p_i s_i \ge p_js_j$ for all other workers $j \ne i$).
Subject to this constraint, we consider all possible numbers of tasks that can be assigned to $i$, and then assign as many tasks as possible to non-attacked workers in order of their proficiency.
Optimality follows from the fact that we exhaustively search possible targets and allocation policies to these, and assign as many tasks as possible to the most effective workers.
\begin{algorithm}[hbt]
	\caption{Optimal allocation of tasks}
	\label{alg:deterministic} 
	\textbf{input:} The set of workers $W$, and their proficiencies $P$\\
	\textbf{return:} The optimal deterministic policy $ s^* $\\
	\begin{algorithmic}[1]
		\State $ u_{max} \leftarrow 0 $
		\For {$ i \in \{1,\ldots,n\} $}
			\For {$ s_i \in \{1,\ldots,m\} $}
				\State $ util \leftarrow 0  $
				\State $ B \leftarrow m - s_i $
				\For { $ j \in \{1,\ldots,n|j\ne i\} $ }
					\State $ s_j \leftarrow \min(\floor*{\frac{p_i}{p_j}  s_i},B) $
					\State $ util \leftarrow util +s_jp_j $
					\State $ B \leftarrow B-s_j $
					\If {$ B = 0 $}
						\State break
					\EndIf
				\EndFor
				\If {$ util >  u_{max}$}
					\State $ u_{max} \leftarrow util $	
					\State $ s^* \leftarrow s $
				\EndIf
			\EndFor
		\EndFor\\
		\Return $ s^* $					
	\end{algorithmic} 
\end{algorithm}

\subsection{Experiments}

We now experimentally consider two questions associated with our problem: 1) what is the impact of the distribution of worker proficiencies on the requester's utility and the number of workers assigned, and 2) what is the difference between optimal randomized and deterministic assignment.
We sample worker proficiencies using two different distributions from which the workers' proficiencies were sampled: a uniform distribution over the $[0.5,1]$ interval, and a power law distribution with $k=0.5$ in which proficiencies are truncated to be in this interval.
We use 100 tasks, unless stated otherwise, and vary the number of workers between 2 and 20.\footnote{When we vary the number of workers, we generate proficiencies incrementally, adding a single worker with a randomly generated proficiency each time.}
For each experiment, we take an average of 20,000 sample runs.

\begin{figure}[htb]
	\centering
	\begin{subfigure}{0.5\textwidth}
		\centering
			\includegraphics[width=1\linewidth]{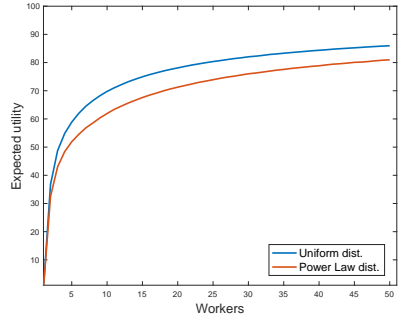}
		\caption{Utility (randomized).}
		\label{subfig:randomizedutility}
	\end{subfigure}%
	\begin{subfigure}{0.5\textwidth}
		\centering
		\includegraphics[width=1\linewidth]{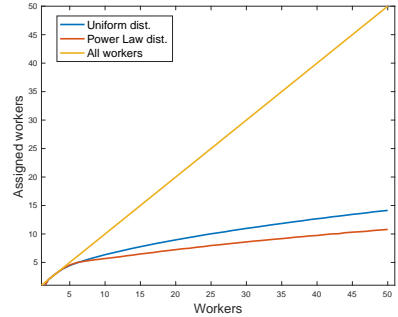}
		\caption{Assignment (randomized).}
		\label{subfig:randomizedass}
	\end{subfigure}
	\begin{subfigure}{0.5\textwidth}
		\centering
			\includegraphics[width=1\linewidth]{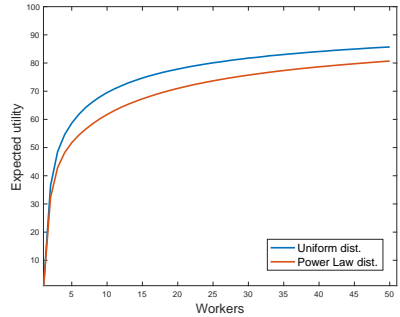}
		\caption{Utility (deterministic).}
		\label{subfig:deterministicutility}
	\end{subfigure}%
	\begin{subfigure}{0.5\textwidth}
		\centering
			\includegraphics[width=1\linewidth]{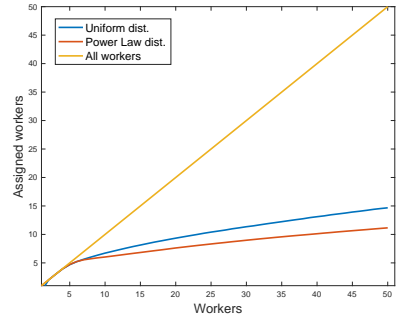}
		\caption{Assignment (deterministic).}
		\label{subfig:deterministicass}
	\end{subfigure}\hfil
	\caption{Comparison between uniform and power law distributions of worker proficiencies.}
	\label{fig:distribution}
\end{figure}
In Figure~\ref{fig:distribution} we compare the uniform and power law distributions in terms of the expected defender utility and the number of workers assigned any tasks in an optimal randomized (Figure~\ref{fig:distribution}a-b) and deterministic (Figure~\ref{fig:distribution}c-d) assignments.
Consistently, under the power law distribution of proficiencies, the defender's utility is lower, and fewer workers are assigned tasks in an optimal assignment.

Next, we experimentally compare the optimal randomized and deterministic assignments in terms of (a) defender's utility, and (b) the number of workers assigned tasks.
In this case, we only show the results for the uniform distribution over worker proficiencies.

It is, of course, well known that the optimal randomized assignment must be at least as good as deterministic (which is a special case), but the key question is by how much.
As Figure~\ref{fig:ratio} shows, the difference is quite small: always below 3\%, and decreasing as we increase the number of tasks from 20 to 200.
Similarly, Figure~\ref{fig:assignmentcomp} suggests that the actual policies are not so different in nature: roughly the same number of most proficient workers are assigned to tasks in both cases.
\begin{figure}
	\centering
	\begin{subfigure}{0.5\textwidth}
		\centering
                  \includegraphics[width=1\textwidth]{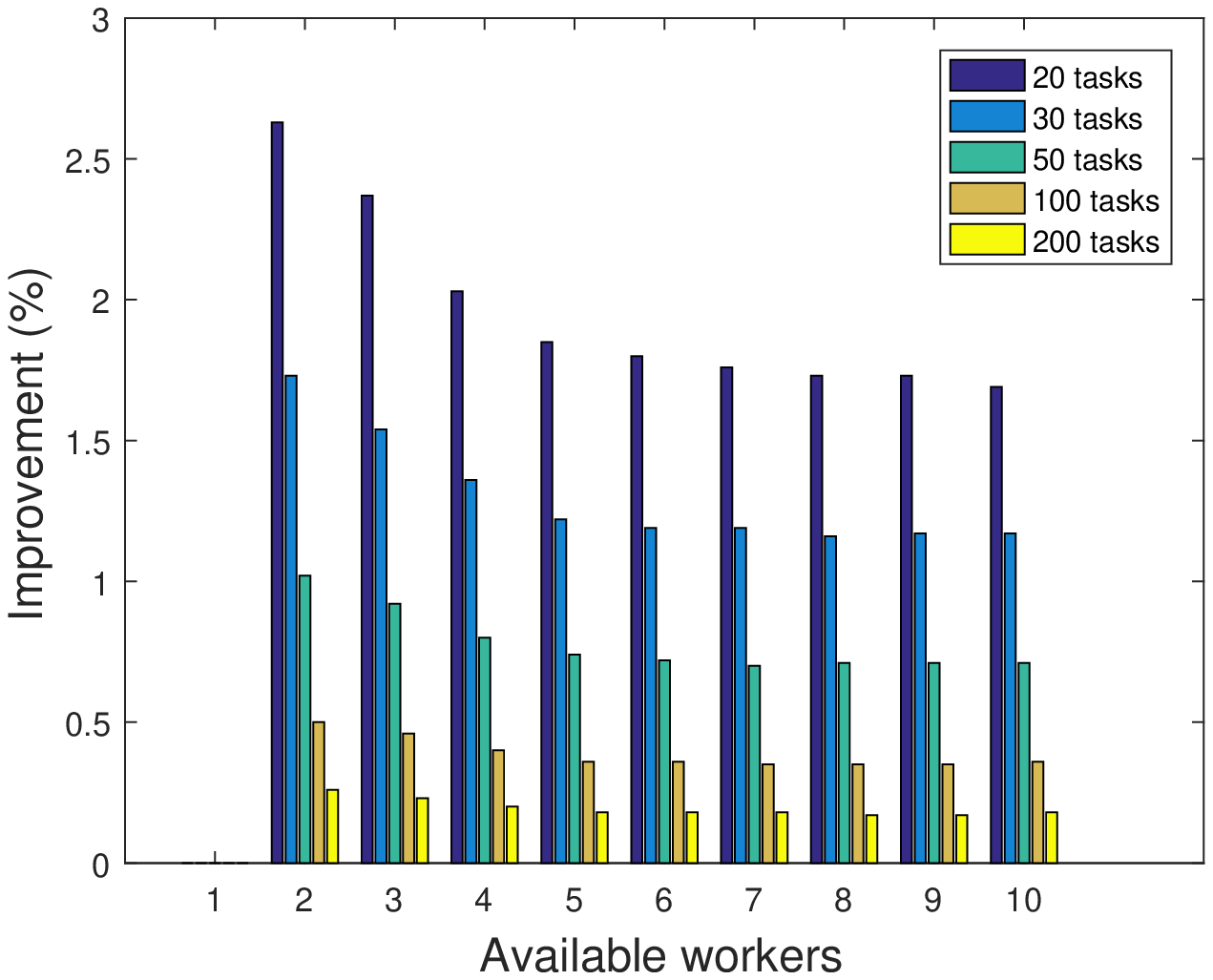}
\caption{Improvement ratio - \\Randomized vs. Deterministic.}
\label{fig:ratio}
	\end{subfigure}%
	\begin{subfigure}{0.5\textwidth}
		\centering
		\includegraphics[width=1\textwidth]{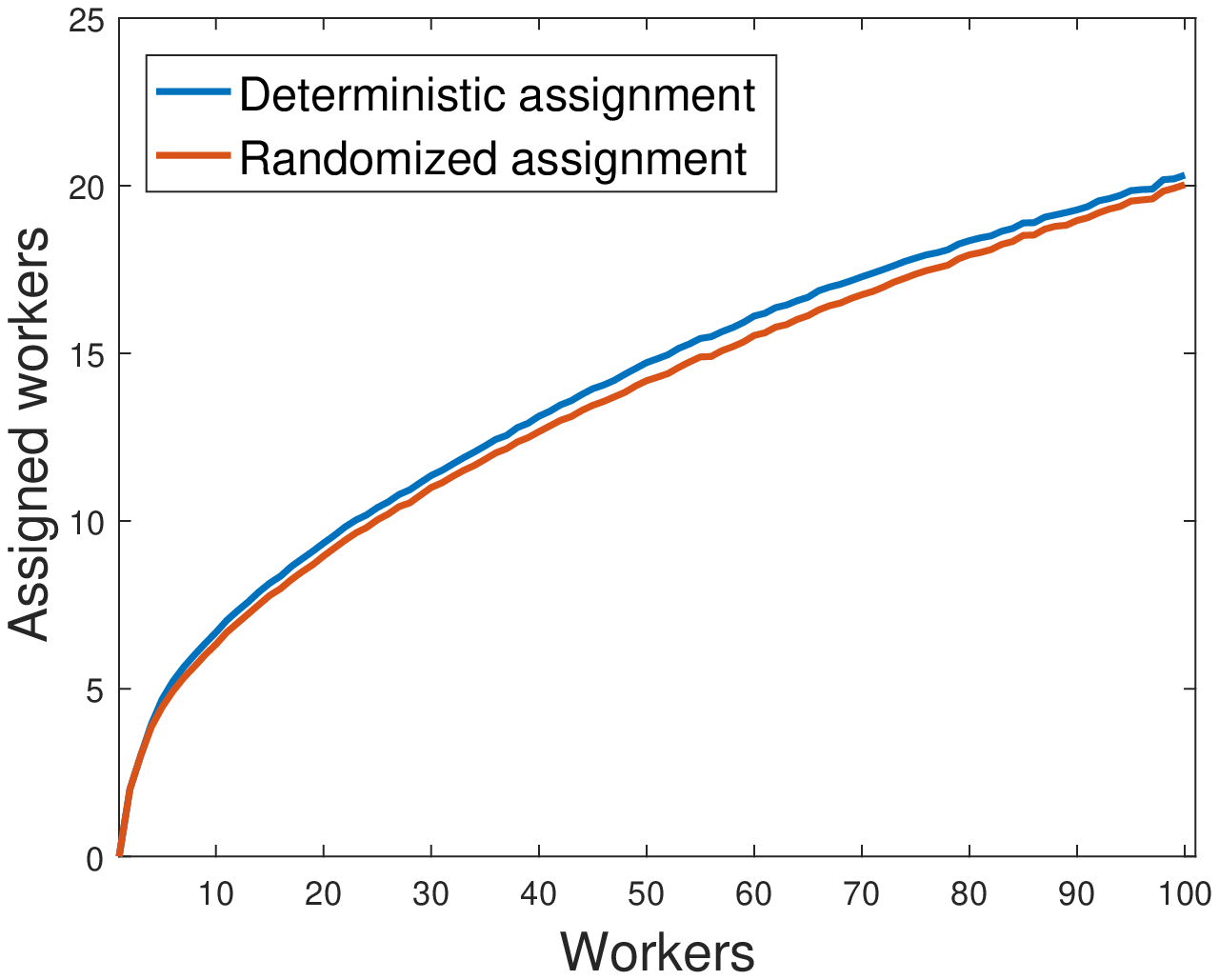}
		\caption{Assigned workers}
		\label{fig:assignmentcomp}
	\end{subfigure}
	\caption{Assignments' comparison}
\end{figure}
The implication of this observation is that from the defender's perspective it is not crucial to know precisely what the adversary observes about the assignment policy: one can safely use the optimal deterministic policy, which is near-optimal even when the attacker only observes the policy and not the actual assignment.
On the other hand, the deterministic assignment is much more robust: in a randomized assignment, if the attacker actually observes the worker which tasks are assigned to, the defender will receive zero utility, whereas an optimal deterministic assignment can achieve a near-optimal utility in any case.

\section{Heterogeneous tasks}


It turns out that the more general problem in which utilities are heterogeneous is considerably more challenging than the case of homogeneous allocation.
First, we show that even if the tasks' utilities slightly different, it may be beneficial to assign the same task to multiple workers. Consider the case of an environment populated with $ 2 $ workers and $ 2 $ tasks. WLOG, we order the tasks by their utility, i.e., $ u_{t_1} < u_{t_2} $. Regardless of the workers' proficiencies, assigning one worker per task will result in an expected utility of $ \min (p_iu_{t_1}, p_ju_{t_2}) $. Still, assigning both worker to $ t_2 $ will result in an expected utility of $ \min (p_iu_{t_2},p_ju_{t_2}) $ which is promised to be higher than the previous one.
Aside from the considerably greater complexity challenge associated in solving problems with heterogeneous utilities alluded to in this example, there is an additional challenge of resolving disagreement among workers, particularly when there are an even number of them.
We leave this issue for future work, and for the moment tackle a restricted problem in which the defender nevertheless assigns only a single worker per task.

\subsection{Randomized strategy}


If we assume that a single worker is assigned to each task, it turns out that we can apply Algorithm~\ref{alg:DoSRand} directly in the case of randomized assignment as well.
To show this, we need to extend Proposition~\ref{prop:rep} to the heterogeneous assignment case; the remaining propositions, with the provision that one worker is assigned per task, do not rely on the fact that tasks are homogeneous and can be extended with minor modifications.
To this end, let $s_{wt}$ be a binary variable which is 1 iff a worker $w$ is assigned to task $t$.
From our assumption, $\sum_{w} s_{wt} = 1$ for each $t$ (since the budget constraint is $m$, we would assign a worker to each task).
Further, define $U = \sum_{t\in T} u_t$.

\begin{prop}\label{prop:heterogeneousRandom}
	Suppose tasks are heterogeneous and one worker is assigned to each task.
Then for any distribution over assignments $q$ and attack strategy $\alpha$, there exists a distribution over $\tilde{S}$, $\lambda$, which results in the same utility.
\end{prop}
\begin{proof}
Fix an attacker strategy $\alpha$, and let $S$ the set of assignments $s$ in which a single worker is assigned to each task.
For any probability distribution $q$ over assignments $s \in S$, the expected utility of the defender is $u_{def}(q,\alpha) = \sum_{w\in W} p_w (1-\alpha_w) \sum_{s \in S} q_s \sum_{t\in T}s_{wt}u_t$.
The expected utility of the defender for a distribution $\lambda$ over $\tilde{S}$ (i.e., over workers) is $u_{def}(\lambda,\alpha) = \sum_{t\in T}u_t\sum_{w\in W} p_w (1-\alpha_w) \lambda_w = U \sum_{w\in W} p_w (1-\alpha_w) \lambda_w$.
Define $\lambda_w = \frac{\sum_{s \in S} q_s \sum_{t\in T}s_{wt}u_t}{U}$.
It suffices to show that $\sum_{w\in W} \lambda_w = 1$.
This follows since $\frac{\sum_{s \in S} q_s \sum_{t\in T}\sum_{w\in W}s_{wt}u_t}{U} = \frac{U\sum_{s \in S} q_s}{U} = 1$ because $ \sum_{w\in W}s_{wt}  = 1 $ and $q$ is a probability distribution.
\end{proof}
Thus, if we constrain the defender to use a single worker per task, we can randomize over workers, rather than full assignments, allowing us to compute a (restricted) optimal randomized assignment in linear time.

\subsection{Deterministic strategy}

We now show that the defender's deterministic allocation problem, denoted \textit{Heterogeneous tasks’ deterministic assignment (HTDA)}, is NP-hard even if we restrict the strategies to assign only a single worker per task.
\begin{prop}
	HTDA is strongly NP-hard even when we assign only one worker per task.
\end{prop}
\begin{proof}
 We reduce from the \textit{Bin packing problem (BP)}, which is a strongly NP-Hard problem. In the bin packing problem, objects of different volumes must be packed into a finite number of bins or containers each of volume $ V $ in a way that minimizes the number of bins used. We define the set of items $ I = \{i_1,\ldots,i_m\} $ with sizes $ \varsigma = \{\varsigma_1,\ldots,\varsigma_n\} $, the volume as $ V = \gamma $, and the set of containers as $ C=\{c_1,\ldots,c_n\} $. The decision problem is deciding if objects will fit into a specified number of bins. Our transformation maps the items to $ m+1 $ tasks with the following utilities $ \{\varsigma_1, \ldots, \varsigma_n, \gamma\} $, the $ n $ containers to $ n+1 $ workers while considering the private case where all the workers have the same proficiency (i.e., $ p_i = p_j $, $ \forall w_i,w_j \in W $). If we started with a YES instance for the BP problem, there is an assignment of items to containers under the volume $ \gamma $. Let $ \mathcal{A} $ be that assignment. Then if $ \mathcal{A}(t_i) = j $, we assign task $ t_i $ to worker $ j $ on HTDA. Also, we assign task $ t_{m+1} $ (with utility $ \gamma $) to worker $ n+1 $. The utility of this task assignment is: $ V + \gamma - \gamma = V $. For the case that with a NO instance for the BP problem, assume in negation that this is a YES instance for the HTDA problem. I.e., there exists an assignment in HTDA such that $ V^* - \gamma^* \geq V $, where $ V^* = \sum_{w_i\in W}\sum_{t\in T} x_{it} u_t  p_{i} $ and $ \gamma^* = \max(\sum_{t\in T} x_{it} u_t  p_{i}) $.  This  implies that $ V + \gamma \geq V^* $. Substituting $ V^* = V+\gamma^* $, we get that  $ V + \gamma - \gamma^* \geq V $, hence, $ \gamma \geq \gamma^*  \geq \sum_{t\in T} x_{it} u_t  p_{i} $, $ \forall w_i \in W $. Note that this contradicts the assumption that this is a YES instance for the HTDA problem. The reduction can clearly be performed in polynomial time. 
\end{proof} 

We propose the following integer program for computing the optimal deterministic strategy for the defender (assuming only one worker is assigned per task):
\begin{subequations}
	\label{E:deterministicHet}
	\begin{align}
	& \max_{s,\gamma} \displaystyle\sum\limits_{w\in W}\sum\limits_{t\in T} s_{wt} u_t  p_{w} - \gamma  \label{objectiveDH}\\
	& s.t.: \displaystyle\sum\limits_{w\in W}\sum\limits_{t\in T} s_{wt} = m \label{cons:allTasksDH}\\
	&\gamma \ge \sum\limits_{t\in T} s_{wt}u_t  p_{w}, \forall w \in W \label{cons:maxDH}\\
        &\sum_w s_{wt} = 1, \forall t\in T\label{cons:singleWorkerTask}\\
         &s_{wt} \in \{0,1\}.
	\end{align}
\end{subequations}
The objective (\ref{objectiveDH}) aims to maximize the defender's expected utility given the adversary's attack (second term). Constraint (\ref{cons:allTasksDH}) ensures that each allocation assigns all the possible tasks among the different workers and Constraint (\ref{cons:maxDH}) validates that the adversary's target is the worker who contributes the most to the defender's expected utility.
Finally, Constraint~\eqref{cons:singleWorkerTask} ensures that only one worker is assigned for each task.




\subsection{Experiments}

This analysis compares the defender's expected utility while using optimal randomized and deterministic strategies when we restrict that only one worker can be assigned to each task. 
We used CPLEX version 12.51 to solve the linear and integer programs above. The simulations were run on a 3.4GHz hyperthreaded 8-core Windows machine with 16 GB RAM.
We generated utilities of different tasks using 6 different uniform distributions: \{[0,0.5],[0,1],[0,5],[0,10],[0,50],[0,100]\}, varied the number of workers between 2 and 15, and considered 15 tasks.
Worker proficiencies were again sampled from the uniform distribution over the [0.5,1] interval.
Results were averages of 1,000 simulation runs.

\begin{figure}
	\centering
	\includegraphics[width=0.72\linewidth]{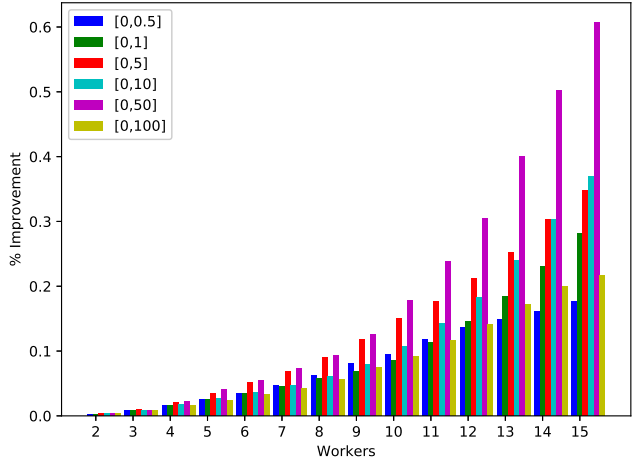}
	\caption{Percentage difference between optimal randomized and deterministic allocations, when we constrain to assign one worker per task and tasks are heterogeneous.}
	\label{fig:compareHet}
\end{figure}
Figure~\ref{fig:compareHet} shows proportion difference between randomized and deterministic allocations for different numbers of workers and distributions from which task utilities are generated.
As we can observe, the difference is remarkably small: in all cases, the gain from using a randomized allocation is below 0.6\%, which is even smaller (by a large margin) than what we had observed in the context of homogeneous tasks.
However, there is an interesting difference we can observe from the homogeneous task setting: now increasing the number of workers considerably increases the advantage of the randomized allocation, whereas when tasks are homogeneous we saw the opposite trend.


\section{Discussion and Conclusions}

We consider the problem of assigning tasks to workers in an
adversarial setting when a worker can be attacked, and their ability
to successfully complete assigned tasks compromised.
In our model, since the defender obtains utility only from correctly
annotated tasks, the nature of the attack is less important; thus, the
attacker can compromise the integrity of the labels reported by the
worker, or simply prevent the worker from completing the tasks
assigned to them.
A key feature of our model is that the attack takes place after the
tasks have been assigned to workers, but has considerable structure in
that exactly one worker is attacked.
Additional structure is imposed by considering two settings: one in
which the attacker only observes the defender's (possibly randomized)
task allocation policy, and the other in which the actual task
assignment decision is known.
We show that the optimal randomized allocation problem in the former
setting (in the sense of Stackelberg equilibrium commitment) can be
found in linear time.
However, our algorithm for optimal deterministic commitment is pseudo-polynomial.
Furthermore, when tasks are heterogeneous, we show that the problem is
more challenging, as it could be optimal to assign multiple workers to
the same task.
If we nevertheless constrain that only one worker is assigned per
task, we can still compute an optimal randomized commitment in linear
time, while deterministic commitment becomes strongly NP-Hard (we
exhibit an integer linear program for the latter problem).

\bibliographystyle{ieeetr}
\bibliography{adversarialC}
\end{document}